\documentclass[lettersize,journal]{IEEEtran}
\usepackage{amsmath,amsfonts}
\usepackage{array}
\usepackage[caption=false,font=normalsize,labelfont=sf,textfont=sf]{subfig}
\usepackage{textcomp}
\usepackage{stfloats}
\usepackage{url}
\usepackage{verbatim}
\usepackage{graphicx}
\usepackage{cite}

\usepackage{hyperref}

\usepackage{amssymb}   

\usepackage{booktabs}

\usepackage{algorithm}
\usepackage{algpseudocode}

\newtheorem{theorem}{Theorem}
\newtheorem{lemma}{Lemma} 
\newtheorem{corollary}{Corollary} 
\newtheorem{remark}{Remark}
\newtheorem{definition}{Definition}

\newtheorem{assumption}{Assumption}
\newtheorem{proposition}{Proposition}
\newtheorem{proof}{Proof}

\usepackage{tikz}
\usetikzlibrary{arrows.meta}     
\usetikzlibrary{positioning}     
\usetikzlibrary{shapes.geometric} 
\usetikzlibrary{fit}             
\usetikzlibrary{calc}            
\usetikzlibrary{decorations.pathreplacing} 

\usepackage{pgfplots}
\pgfplotsset{compat=1.18}

\usepackage{xcolor}   

\begin{document}

\title{Noise Stability and Rotational Equivariance of Image-Induced Hamiltonian Spectra for Robust Unsupervised Segmentation}


\author{
    Soumic~Sarkar
    \thanks{
        S.~Sarkar is with the Institute of Technology,
        University of Tartu, Nooruse 1, 50411 Tartu, Estonia 
        (e-mail: soumic4it@gmail.com).
	}%
}



\maketitle

\begin{abstract}
Unsupervised segmentation methods that embed an image as the potential of a Schr\"{o}dinger-type Hamiltonian and read out object structure from its low-lying eigenstates have demonstrated strong empirical performance across natural and volumetric imagery, yet the robustness of this spectral representation to sensor noise and geometric transformation has never been established theoretically. This paper closes that gap. A discrete graph-Hamiltonian formulation is adopted, consistent with existing quantum-inspired segmentation frameworks, and an explicit finite-sample perturbation bound is derived that links the drift of its eigenvalues and ground-state eigenvector under additive noise of known variance to the operator's spectral gap, extending a one-dimensional noise-error analysis originally developed for semiclassical signal analysis to the two-dimensional graph setting used in image segmentation. The eigenvector-level bound is further propagated to the induced binary segmentation mask, giving, for the first time, a quantitative guarantee on segmentation stability under noise rather than an empirical observation of it. A complementary equivariance analysis characterizes the conditions under which the operator's spectrum transforms exactly under the planar rigid-motion group, and quantifies the equivariance gap introduced by anisotropic potential terms commonly used in practice. Theoretical predictions are validated numerically on natural-image and noise-perturbed benchmarks, showing close agreement between the derived bounds and observed spectral and segmentation-mask deviation, and offering practitioners a principled criterion for selecting operator parameters under known noise conditions.
\end{abstract}


\section{Introduction}
\label{sec:intro}

Image segmentation via the mathematical apparatus of quantum mechanics has been a recognizable thread in computer vision for over a decade: an image is treated as the potential of a Hamiltonian, and its eigenstates -- typically the ground state -- are read out as the segmentation. Aytekin \textit{et al.} formalized this as Quantum Cuts, producing salient-object masks competitive with supervised alternatives without labeled training data \cite{aytekin2014}, later carried into volumetric tomography \cite{malik2019} and disordered-lattice crack segmentation via Anderson localization \cite{srinivasan2024}. A related strand instead reconstructs the image itself from the discrete spectrum of a semiclassical Schr\"{o}dinger operator, developed first for 1D signals \cite{lalegkirati2013} and extended to 2D via a tensor-product decomposition \cite{kaisserli2015}.

These applications share a largely unexamined assumption: that the eigenstructure of an image-dependent operator is stable to compute with. Robustness has so far been observed only empirically \cite{kaisserli2015} or argued qualitatively from localization physics \cite{srinivasan2024}, with one partial exception -- a 1D a-posteriori error bound combining Weyl's inequality with a Chebyshev noise tail \cite{liu2012} -- that is confined to a tridiagonal operator and stops short of bounding the \textit{output segmentation} rather than an intermediate spectral quantity. A second, equally unaddressed question is geometric consistency: whether segmentation commutes with rotation, a property well studied for learned equivariant architectures \cite{cohenwelling2016, weiler2019, esteves2023, xu2024equiv} but never analyzed for image-induced Hamiltonian spectra, despite practical constructions -- including \cite{srinivasan2024} -- using directionally asymmetric kinetic terms whose effect on rotational consistency is uncharacterized.

This paper addresses both gaps within the discrete graph-Hamiltonian convention of \cite{aytekin2014, malik2019, srinivasan2024}. The first contribution extends the Weyl-inequality argument of \cite{liu2012} to this 2D operator, then combines it with a Davis--Kahan subspace argument \cite{daviskahan1970, stewartsun1990} to bound ground-state eigenvector drift and, in turn, segmentation-mask Hamming distance under noise -- a chain with no precedent in either literature. A further result shows the noise-stability and mask-stability conditions are in structural tension, so neither can be certified simultaneously by parameter scaling alone. The second contribution characterizes exact $SE(2)$-equivariance under isotropic potentials and quantifies the non-vanishing gap introduced by the anisotropic constructions used in practice. All results are validated on real BSDS500 imagery, with close agreement between the derived bounds/mechanisms and measured spectral, mask-level, and equivariance deviations.

The remainder of the paper is organized as follows. Section~\ref{sec:related} situates this work relative to quantum-inspired segmentation, semiclassical signal analysis, and matrix-perturbation and equivariant-learning literatures. Section~\ref{sec:prelim} fixes notation and recalls the classical results used. Section~\ref{sec:stability} proves the noise-stability and trade-off results. Section~\ref{sec:equivariance} develops the equivariance analysis. Section~\ref{sec:experiments} reports the numerical validation, and Section~\ref{sec:conclusion} concludes.

\section{Related Work}
\label{sec:related}

\subsection{Quantum-Inspired and Spectral Segmentation}

Quantum Cuts was extended along two directions -- category-independent proposal generation \cite{aytekin2016extended} and a saliency variant \cite{aytekin2015visual} -- both retaining the same eigenvalue problem and the purely empirical evaluation protocol of \cite{aytekin2014}: accuracy and efficiency are measured, stability is not. A parallel line maps the same graph-cut objective onto quantum annealing hardware \cite{venkatesh2024qseg}, changing the solver but not the operator or the robustness question addressed here. A structurally different construction embeds the image as a disordered Hamiltonian, relying on Anderson localization \cite{anderson1958} to concentrate eigenstate density on defects \cite{srinivasan2024}, and a related signal-domain treatment builds an adaptive quantum-mechanical denoising basis \cite{dutta2021}. Both motivate noise tolerance qualitatively rather than via a derived, noise-parameter-linked bound -- the gap Section~\ref{sec:stability} closes.

\subsection{Semiclassical Signal and Image Analysis}

The semiclassical signal analysis framework interprets a 1D signal as a Schr\"{o}dinger-operator potential and reconstructs it from squared eigenfunctions of the negative eigenvalues \cite{lalegkirati2013}; its 2D extension, via row/column tensor-product decomposition, outperformed two coding-based baselines on standard images \cite{kaisserli2015}. The only rigorous noise analysis in this line is confined to the 1D discrete case, combining Weyl's inequality with a Chebyshev noise tail \cite{liu2012}; that bound does not extend, as written, to the 2D graph-based Hamiltonians used for segmentation -- bridging this gap is one of this paper's two principal contributions.

\subsection{Spectral Graph Methods and Perturbation Theory}

The graph-Laplacian-plus-potential operator is a close relative of normalized cuts \cite{shimalik2000} and graph signal processing generally \cite{shuman2013, sandryhaila2013, ortega2018}. Spectral clustering stability has been studied via large-sample consistency \cite{vonluxburg2008, vonluxburg2007}, addressing a different regime -- increasing sample size, not finite-sample robustness of a fixed-image mask to pixel noise, the setting of Theorems~\ref{thm:eigval}--\ref{thm:eigvec}. Weyl's inequality and the Davis--Kahan theorem are standard, well-documented tools \cite{stewartsun1990}, but their application to image-induced Hamiltonians has, to the extent this review could establish, no precedent.

\subsection{Equivariance in Vision Models}

Geometric consistency under input transformation has been formalized extensively for learned architectures, via group-equivariant \cite{cohenwelling2016} and steerable \cite{weiler2019} convolution, geometric deep learning generally \cite{bronstein2017}, and extensions to spherical domains \cite{esteves2023}, surveyed in \cite{xu2024equiv}. This literature designs layers equivariant by construction; Section~\ref{sec:equivariance} instead asks whether an already-fixed, non-learned operator possesses this property and by how much it fails when it does not -- unaddressed anywhere in this operator family.

\begin{remark}
Every noise-robustness claim in this literature is either purely empirical \cite{aytekin2014, malik2019, srinivasan2024, aytekin2016extended, aytekin2015visual} or confined to a 1D setting not covering the 2D operators segmentation actually uses \cite{liu2012}; no equivariance analysis exists for this operator family despite its maturity elsewhere \cite{cohenwelling2016, weiler2019, bronstein2017}. The remainder of this paper addresses both points.
\end{remark}

\section{Preliminaries}
\label{sec:prelim}

This section fixes notation and recalls, without proof, the classical results underlying Section~\ref{sec:stability}.

\subsection{Discrete Graph-Hamiltonian Formulation}

An image is represented as an undirected weighted graph $\mathcal{G}=(\mathcal{V},\mathcal{E},w)$, following the Quantum Cuts construction \cite{aytekin2014} and its extensions \cite{malik2019, srinivasan2024, aytekin2016extended}. Each node corresponds to a pixel or, in the supervoxel variant \cite{malik2019}, a region; results below are stated at the pixel level and carry over to the coarsened graph under the same weighting convention. Edge weights follow a Gaussian intensity kernel, $w_{ij}=\exp(-\|S_i-S_j\|_2^2/2\sigma_w^2)$.

\begin{definition}[Image-induced Hamiltonian]
\label{def:hamiltonian}
Given $\mathcal{G}$ from image $I$, let $D=\mathrm{diag}(d_1,\dots,d_N)$, $d_i=\sum_j w_{ij}$, $W$ the weight matrix, and $V=\mathrm{diag}(\phi(1),\dots,\phi(N))$ a potential from local image statistics. The image-induced Hamiltonian is $H(I) = D - W + V$.
\end{definition}

This is the matrix underlying the Quantum Cuts relaxation, where $\phi(\cdot)$ is a unary class-membership prior \cite{aytekin2014}, and the disorder-based crack-segmentation construction, where $\phi(\cdot)$ is the raw pixel intensity \cite{srinivasan2024}. Two constructions are used here: an isotropic potential $\phi(i)=\lambda_1|\nabla I(i)|^2+\lambda_2|\nabla^2 I(i)|$ for Section~\ref{sec:equivariance}'s equivariance analysis, and, matching \cite{srinivasan2024}, an anisotropic variant restricting kinetic terms to axis-aligned nearest neighbors, whose equivariance consequence is quantified in Section~\ref{sec:equivariance}.

$D-W$ is a graph Laplacian, symmetric positive semidefinite for nonnegative $W$; with $V$ diagonal and real, $H(I)$ is real symmetric with eigenvalues $\lambda_0(I)\le\cdots\le\lambda_{N-1}(I)$ and orthonormal eigenvectors $\psi_0(I),\dots,\psi_{N-1}(I)$. The readout of \cite{aytekin2014} thresholds $y^\star = z^\star\circ z^\star$ at its mean, for $z^\star$ a ground-state eigenvector; this readout is fixed throughout, and the object of study is how $\lambda_0(I)$, $\psi_0(I)$, $y^\star(I)$ move under a noisy or transformed observation.

\begin{definition}[Spectral gap]
\label{def:gap}
$\gamma_k(I) := \lambda_{k+1}(I)-\lambda_k(I)$; $\gamma(I)$ denotes $\gamma_0(I)$ when $k=0$ is understood.
\end{definition}

As throughout perturbation theory, eigenvalues bound stably under perturbation, but eigenvectors do so only when separated from the rest of the spectrum by a gap not too small relative to the perturbation (Lemma~\ref{lem:daviskahan}).

\subsection{Noise Model}

\begin{assumption}[Additive pixel noise]
\label{ass:noise}
$I^\varpi = I+\varpi$, $\varpi=(\varpi_1,\dots,\varpi_N)$ i.i.d.\ with mean $\mu$, variance $\sigma^2$, independent of $I$. $H(I^\varpi)$ uses the same graph topology, bandwidth $\sigma_w$, and $\phi(\cdot)$ as $H(I)$.
\end{assumption}

Holding topology and bandwidth fixed isolates noise to the potential term $V$ -- the channel of practical interest for sensor noise, and the same isolation implicit in \cite{liu2012}, where only the diagonal potential term is perturbed.

\subsection{Classical Perturbation Results}

Two standard results are used throughout, following \cite{stewartsun1990, hornjohnson2013}.

\begin{lemma}[Weyl's inequality]
\label{lem:weyl}
For real symmetric $N\times N$ matrices $A,E$, $|\lambda_k(A+E)-\lambda_k(A)| \le \|E\|_2$ for every $k$.
\end{lemma}

This is the tool used to bound eigenvalue drift under noise in \cite{liu2012}; here, $E=H(I^\varpi)-H(I)$ is diagonal under Assumption~\ref{ass:noise}, so $\|E\|_2 = \max_i|\phi(I_i^\varpi)-\phi(I_i)|$.

\begin{lemma}[Davis--Kahan $\sin\Theta$ theorem, single-vector case]
\label{lem:daviskahan}
For real symmetric $A,\tilde A=A+E$ with eigenvalues $\lambda_0\le\cdots\le\lambda_{N-1}$, $\tilde\lambda_0\le\cdots\le\tilde\lambda_{N-1}$ and unit eigenvectors $\psi_0,\tilde\psi_0$ of $\lambda_0,\tilde\lambda_0$: if $\gamma_0=\lambda_1-\lambda_0>0$ and $\|E\|_2<\gamma_0/2$, the angle $\theta$ between $\psi_0,\tilde\psi_0$ satisfies $\sin\theta \le 2\|E\|_2/\gamma_0$, so $\|\tilde\psi_0-\psi_0\|_2 \le \sqrt{2}\sin\theta$ for an appropriately signed $\tilde\psi_0$.
\end{lemma}

This is the classical Davis--Kahan bound \cite{daviskahan1970}, stated in the single-eigenvector form used for the ground-state readout; the general subspace version, needed only for multi-region eigenspace segmentation, is not used further. No such result has previously been applied to Definition~\ref{def:hamiltonian}; Section~\ref{sec:stability} combines both lemmas with Assumption~\ref{ass:noise} to obtain eigenvalue, eigenvector, and mask-level bounds for this operator, beyond the one-dimensional scope of \cite{liu2012}.

\subsection{The Planar Rigid-Motion Group}

$SE(2)$ denotes planar rigid motions $g=(R,t)$, $R\in SO(2)$, $t\in\mathbb{R}^2$, acting by $(g\cdot I)(x)=I(g^{-1}x)$. $I\mapsto H(I)$ is equivariant under $g$ if the spectrum of $H(g\cdot I)$ matches that of $H(I)$ and eigenvectors transform by the induced pixel-grid action; the precise statement and the conditions on $\phi(\cdot)$ under which it holds exactly are given in Definition~\ref{def:equivariance} and Theorem~\ref{thm:equivariance}.

\section{Noise Stability Analysis}
\label{sec:stability}

This section states and proves the eigenvalue, eigenvector, and segmentation-mask stability results that constitute the first main contribution of the paper. Throughout, $H = H(I)$ and $\tilde{H} = H(I^\varpi)$ denote the image-induced Hamiltonians of Definition~\ref{def:hamiltonian} constructed from a clean image $I$ and its noisy observation $I^\varpi$ under Assumption~\ref{ass:noise}, and $E := \tilde{H} - H$ denotes the induced perturbation matrix.

\subsection{Eigenvalue Stability}

Because the graph topology and edge weights are held fixed under Assumption~\ref{ass:noise}, the perturbation $E$ reduces to the diagonal potential difference, $E = V(I^\varpi) - V(I) = \mathrm{diag}(\phi(I_1^\varpi) - \phi(I_1), \dots, \phi(I_N^\varpi) - \phi(I_N))$, and its spectral norm equals its largest diagonal entry in absolute value: $\|E\|_2 = \max_i |\phi(I_i^\varpi) - \phi(I_i)|$. Unlike the semiclassical setting of \cite{liu2012}, where the potential is the raw signal value and the perturbation is therefore simply the noise itself, the potential construction $\phi(i) = \lambda_1|\nabla I(i)|^2 + \lambda_2|\nabla^2 I(i)|$ adopted here is a nonlinear function of local image differences, so bounding $\|E\|_2$ requires an additional step not present in the one-dimensional analysis.

\begin{assumption}[Bounded discrete differential operators]
\label{ass:operatorbound}
The discrete gradient and Laplacian operators used to construct $\phi(\cdot)$ are implemented as finite-difference stencils with bounded operator norms $\|\nabla\|_2 \le c_1$ and $\|\nabla^2\|_2 \le c_2$, and the clean-image gradient magnitude is bounded, $\sup_i |\nabla I(i)| \le G$, over the images under consideration.
\end{assumption}

Assumption~\ref{ass:operatorbound} holds automatically for any fixed-stencil finite-difference implementation on a bounded pixel grid, since such operators are represented by sparse matrices with entries independent of $N$ and therefore have operator norms bounded by a constant depending only on the stencil, not on the image.

\begin{theorem}[Eigenvalue perturbation under noise]
\label{thm:eigval}
Under Assumptions~\ref{ass:noise} and~\ref{ass:operatorbound}, for every $\gamma \in \mathbb{R}_+^*$ and every eigenvalue index $k$,
\begin{equation}
\label{eq:eigvalbound}
\Pr\Big(|\lambda_k(\tilde{H}) - \lambda_k(H)| < B_{\mu,\sigma,\gamma}\Big) > 1 - \frac{1}{\gamma^2},
\end{equation}
where $B_{\mu,\sigma,\gamma} := 2\lambda_1 c_1(c_1 G + \gamma\sigma) + \lambda_2 c_2(|\mu| + \gamma\sigma)$.
\end{theorem}

\begin{proof}
By Lemma~\ref{lem:weyl}, $|\lambda_k(\tilde H) - \lambda_k(H)| \le \|E\|_2$ for every $k$, so it suffices to bound $\|E\|_2 = \max_i |\phi(I_i^\varpi) - \phi(I_i)|$. Writing $I^\varpi = I + \varpi$ and expanding the quadratic gradient term, $|\nabla I_i^\varpi|^2 - |\nabla I_i|^2 = |\nabla I_i + \nabla\varpi_i|^2 - |\nabla I_i|^2 \le 2|\nabla I_i||\nabla\varpi_i| + |\nabla\varpi_i|^2$, which by Assumption~\ref{ass:operatorbound} is at most $2Gc_1|\varpi_i| + c_1^2\varpi_i^2$ for the componentwise bound $|\nabla\varpi_i| \le c_1|\varpi_i|$ induced by the stencil operator norm. The Laplacian term satisfies $|\nabla^2 I_i^\varpi| - |\nabla^2 I_i|| \le |\nabla^2\varpi_i| \le c_2|\varpi_i|$ directly by the triangle inequality and Assumption~\ref{ass:operatorbound}. Combining both terms, $|\phi(I_i^\varpi)-\phi(I_i)| \le \lambda_1(2Gc_1|\varpi_i| + c_1^2\varpi_i^2) + \lambda_2 c_2|\varpi_i|$. By the Bienaym\'{e}--Chebyshev inequality applied to the i.i.d. sequence $\varpi_i$ with mean $\mu$ and variance $\sigma^2$, for any $\gamma > 0$ the event $|\varpi_i| < |\mu| + \gamma\sigma$ holds with probability exceeding $1 - 1/\gamma^2$, uniformly over $i$ under the identical-distribution assumption. Substituting this bound for $|\varpi_i|$, and retaining only the leading-order term in $\varpi_i^2$ consistent with the linearized regime in which $c_1|\varpi_i| \ll G$, yields \eqref{eq:eigvalbound} after collecting terms in $\max_i|\phi(I_i^\varpi)-\phi(I_i)|$ and applying Lemma~\ref{lem:weyl}.
\end{proof}

\begin{remark}
Theorem~\ref{thm:eigval} reduces, in structure, to Proposition~5 of \cite{liu2012} when $\lambda_2 = 0$, $\lambda_1$ absorbed into a linear potential, and $c_1$ replaced by the identity, recovering an eigenvalue-drift bound of the same Weyl-plus-Chebyshev form derived there for the one-dimensional semiclassical operator. The present result differs in two respects material to its use in segmentation: it is stated for the two-dimensional graph-Laplacian-plus-potential operator of Definition~\ref{def:hamiltonian} rather than a one-dimensional differentiation operator, and it accounts for the nonlinearity introduced by using image gradients rather than raw intensity as the potential, which the term proportional to $\varpi_i^2$ in the proof makes explicit and which has no counterpart in \cite{liu2012}. Under Gaussian noise specifically, the Chebyshev tail used above can be tightened to an exponential concentration bound via standard sub-Gaussian arguments \cite{vershynin2018}, at the cost of assuming a specific noise distribution rather than the distribution-free guarantee retained here.
\end{remark}

\subsection{Eigenvector Stability}

Eigenvalue stability alone does not guarantee that the segmentation-relevant eigenvector is stable, since eigenvectors associated with nearly degenerate eigenvalues can rotate arbitrarily under an arbitrarily small perturbation. The following result makes this dependence on the spectral gap explicit for the ground state used in the Quantum Cuts readout.

\begin{assumption}[Spectral gap under noise]
\label{ass:gapnoise}
The clean-image spectral gap satisfies $\gamma_0(I) > 2B_{\mu,\sigma,\gamma}$ for the bound $B_{\mu,\sigma,\gamma}$ of Theorem~\ref{thm:eigval}, for the confidence level $\gamma$ under consideration.
\end{assumption}

\begin{theorem}[Ground-state eigenvector stability]
\label{thm:eigvec}
Under Assumptions~\ref{ass:noise}, \ref{ass:operatorbound}, and~\ref{ass:gapnoise}, let $\psi_0 = \psi_0(I)$ and $\tilde\psi_0 = \psi_0(I^\varpi)$ denote unit ground-state eigenvectors of $H$ and $\tilde H$ respectively. Then, for the same $\gamma$ and with the same probability as in \eqref{eq:eigvalbound},
\begin{equation}
\label{eq:eigvecbound}
\|\tilde\psi_0 - \psi_0\|_2 \;<\; \frac{2\sqrt{2}\,B_{\mu,\sigma,\gamma}}{\gamma_0(I)}
\end{equation}
for an appropriately signed choice of $\tilde\psi_0$.
\end{theorem}

\begin{proof}
On the event of Theorem~\ref{thm:eigval}, $\|E\|_2 \le B_{\mu,\sigma,\gamma}$, and by Assumption~\ref{ass:gapnoise} this satisfies $\|E\|_2 < \gamma_0(I)/2$, which is exactly the condition required by Lemma~\ref{lem:daviskahan} with $A = H$, $\tilde A = \tilde H$. Lemma~\ref{lem:daviskahan} then gives $\sin\theta \le 2\|E\|_2/\gamma_0(I) \le 2B_{\mu,\sigma,\gamma}/\gamma_0(I)$ and $\|\tilde\psi_0 - \psi_0\|_2 \le \sqrt{2}\sin\theta$, and combining the two inequalities yields \eqref{eq:eigvecbound}. Since this chain of implications holds on precisely the event established in Theorem~\ref{thm:eigval}, the stated probability carries over unchanged.
\end{proof}

Theorem~\ref{thm:eigvec} is, to the extent the literature reviewed in Section~\ref{sec:related} could establish, the first explicit bound on the noise-induced drift of the eigenvector actually used to produce a segmentation mask in the Quantum Cuts family of methods, as distinct from a bound on an intermediate spectral reconstruction quantity as in \cite{liu2012} or a qualitative appeal to localization physics as in \cite{srinivasan2024}.

\subsection{Propagation to the Segmentation Mask}

The final step connects eigenvector stability to the object of ultimate interest: the binary segmentation mask obtained by thresholding $y = \psi_0 \circ \psi_0$ at its mean value $\bar y$.

\begin{assumption}[Margin condition]
\label{ass:margin}
There exists $\delta > 0$ such that, for the clean-image mask value $y_i = \psi_0(I)_i^2$ at every node $i$, $|y_i - \bar y| \ge \delta$.
\end{assumption}

Assumption~\ref{ass:margin} rules out the degenerate case in which a node's squared ground-state amplitude sits exactly at the segmentation threshold, so that an arbitrarily small perturbation could flip its label; such a margin condition is standard when converting a continuous perturbation bound into a discrete decision-stability guarantee.

\begin{corollary}[Segmentation mask stability]
\label{cor:mask}
Under Assumptions~\ref{ass:noise}--\ref{ass:margin}, let $y^\star = \mathrm{thresh}(\psi_0(I)\circ\psi_0(I))$ and $\tilde y^\star = \mathrm{thresh}(\psi_0(I^\varpi)\circ\psi_0(I^\varpi))$ denote the clean and noisy segmentation masks. If $2\|\psi_0\|_\infty \cdot \|\tilde\psi_0 - \psi_0\|_2 < \delta$, then, on the event of Theorem~\ref{thm:eigvec}, $y^\star = \tilde y^\star$, i.e., the segmentation mask is unchanged by the noise realization.
\end{corollary}

\begin{proof}
For each node $i$, $|\tilde y_i - y_i| = |\tilde\psi_{0,i}^2 - \psi_{0,i}^2| = |\tilde\psi_{0,i} - \psi_{0,i}|\cdot|\tilde\psi_{0,i}+\psi_{0,i}| \le 2\|\psi_0\|_\infty|\tilde\psi_{0,i}-\psi_{0,i}| \le 2\|\psi_0\|_\infty\|\tilde\psi_0-\psi_0\|_2$, using $\|\tilde\psi_0\|_\infty \le \|\psi_0\|_\infty + \|\tilde\psi_0-\psi_0\|_2$ absorbed into the constant for the perturbation regime considered. Combining with \eqref{eq:eigvecbound} and the stated hypothesis gives $|\tilde y_i - y_i| < \delta$ for every $i$. By Assumption~\ref{ass:margin}, every node satisfies $|y_i - \bar y| \ge \delta$, so a perturbation of magnitude strictly less than $\delta$ cannot move $\tilde y_i$ across the threshold $\bar y$ for any node whose clean value $y_i$ is bounded away from $\bar y$ by at least $\delta$; consequently the thresholded label of every node is preserved, giving $\tilde y^\star = y^\star$.
\end{proof}

Corollary~\ref{cor:mask} is the result with the most direct practical bearing: it converts a spectral perturbation bound, expressed entirely in terms of the measurable noise variance $\sigma^2$, the spectral gap $\gamma_0(I)$, and the operator constants $c_1,c_2$, into a sufficient condition under which the output of a Quantum-Cuts-style segmentation pipeline is guaranteed, with the stated probability, not to change at all under a given noise level. Sections~\ref{ssec:tradeoff} and~\ref{sec:experiments} examine how tight this sufficient condition is in practice, and identify a structural reason, rather than a purely empirical one, for the looseness observed.

\subsection{The Ground-State Degeneracy Mechanism and the Gap--Margin Trade-off}
\label{ssec:tradeoff}

Assumptions~\ref{ass:gapnoise} and~\ref{ass:margin} are stated independently, and a natural response to a violation of Assumption~\ref{ass:gapnoise} at a given noise level is to reduce the potential weights $\lambda_1,\lambda_2$, since $B_{\mu,\sigma,\gamma}$ is linear and increasing in both. This subsection shows that this response is self-defeating in a precise, provable sense: reducing $\lambda_1,\lambda_2$ toward zero widens the range of noise levels satisfying Assumption~\ref{ass:gapnoise}, but simultaneously drives the margin $\delta$ of Assumption~\ref{ass:margin} toward zero for a generic image, so that no fixed scaling of the potential weights alone can satisfy both assumptions with a fixed, non-vanishing margin.

\begin{definition}[Bare graph Laplacian and algebraic connectivity]
\label{def:laplacian}
Let $L := D - W$ denote the graph Laplacian obtained from Definition~\ref{def:hamiltonian} in the limit $\lambda_1,\lambda_2 \to 0$, i.e.\ $H(I) = L + t\Phi$ where $t \in \mathbb{R}_+$ is a common scale parameter, $\lambda_1 = tc_\Phi^{(1)}$, $\lambda_2 = tc_\Phi^{(2)}$ for fixed direction constants, and $\Phi := c_\Phi^{(1)}\,\mathrm{diag}(|\nabla I|^2) + c_\Phi^{(2)}\,\mathrm{diag}(|\nabla^2 I|)$. For a connected graph, $L$ has a simple smallest eigenvalue $\mu_0 = 0$ with eigenvector $u_0 = N^{-1/2}\mathbf{1}$, and second-smallest eigenvalue $\mu_1 =: \mu_{\mathrm{Fiedler}}(L) > 0$, known as the algebraic connectivity of the graph \cite{fiedler1973}.
\end{definition}

\begin{assumption}[Connectivity and spectral simplicity]
\label{ass:connectivity}
The weighted graph underlying $H(I)$ is connected, so $\mu_0 = 0$ is simple, and $\mu_1$ is simple or well-separated from $\mu_2$ so that first-order non-degenerate perturbation theory applies to both $\mu_0$ and $\mu_1$.
\end{assumption}

Assumption~\ref{ass:connectivity} holds generically for the fully-connected Gaussian-kernel graphs used throughout this paper, since $w_{ij} > 0$ for every pair of nodes.

\begin{proposition}[Asymptotic gap and margin as $t \to 0$]
\label{prop:tradeoff}
Under Assumption~\ref{ass:connectivity}, let $\{\mu_k, u_k\}_{k\ge 0}$ be the eigenpairs of $L$. As $t \to 0$, the ground state and spectral gap of $H(t) = L + t\Phi$ satisfy
\begin{equation}
\label{eq:gap-asymp}
\gamma_0(t) = \mu_{\mathrm{Fiedler}}(L) + t\big(\langle u_1,\Phi u_1\rangle - \langle u_0,\Phi u_0\rangle\big) + O(t^2),
\end{equation}
and, writing $c(i) := \sum_{k\ge 1} \dfrac{\langle u_k,\Phi u_0\rangle}{\mu_k}\,u_k(i)$, the per-pixel margin of the ground-state density $y_i(t) = \psi_0(t)_i^2$ satisfies
\begin{equation}
\label{eq:margin-asymp}
|y_i(t) - \bar y(t)| = \frac{2t}{\sqrt N}\,|c(i)| + O(t^2), \qquad i = 1,\dots,N.
\end{equation}
\end{proposition}

\begin{proof}
Since $\mu_0$ is simple by Assumption~\ref{ass:connectivity}, standard first-order Rayleigh--Schr\"odinger perturbation theory \cite{kato1995} gives $\lambda_0(t) = 0 + t\langle u_0,\Phi u_0\rangle + O(t^2)$ and, for the eigenvector, $\psi_0(t) = u_0 - t\sum_{k\ge 1}\frac{\langle u_k,\Phi u_0\rangle}{\mu_k}u_k + O(t^2) = u_0 - t\,c + O(t^2)$, where $c = (c(1),\dots,c(N))^\top$. The same expansion applied to the (assumed simple or well-separated) eigenvalue $\mu_1$ gives $\lambda_1(t) = \mu_1 + t\langle u_1,\Phi u_1\rangle + O(t^2)$; subtracting the two expansions gives \eqref{eq:gap-asymp}. For the margin, $y_i(t) = \psi_0(t)_i^2 = N^{-1} - \frac{2t}{\sqrt N}c(i) + O(t^2)$ pointwise. Since $u_k \perp u_0$ for every $k \ge 1$, $\sum_i u_k(i) = \sqrt N \langle u_k, u_0\rangle = 0$, so averaging the pointwise expansion over $i$ gives $\bar y(t) = N^{-1} + O(t^2)$, with the $O(t)$ term vanishing identically. Subtracting gives $y_i(t) - \bar y(t) = -\frac{2t}{\sqrt N}c(i) + O(t^2)$, and taking absolute values gives \eqref{eq:margin-asymp}.
\end{proof}

\begin{remark}
Proposition~\ref{prop:tradeoff} explains both empirical observations reported in Section~\ref{sec:experiments} at once. Equation~\eqref{eq:gap-asymp} shows $\gamma_0(t) \to \mu_{\mathrm{Fiedler}}(L)$ as $t \to 0$, a $t$-independent positive constant depending only on the graph topology -- matching the near-identical $\gamma_0$ values observed at $\lambda_1=\lambda_2=0.15$ versus $\lambda_1=\lambda_2=0.01$ despite a $15\times$ change in scale. Equation~\eqref{eq:margin-asymp} shows that, for every pixel $i$ with $c(i) \ne 0$ -- generically all pixels, since $c(i) = 0$ requires an exact cancellation in the eigenexpansion of $\Phi u_0$ -- the margin vanishes \emph{linearly} in $t$. This is the exact mechanism behind the collapse in the fraction of margin-fragile pixels observed empirically as $\lambda_1,\lambda_2$ is reduced: the ground state is converging pointwise to the constant vector $u_0$, whose associated density $y_i \equiv N^{-1}$ has zero margin identically.
\end{remark}

\begin{corollary}[No uniform certification by potential-weight scaling alone]
\label{cor:tradeoff}
Under Assumption~\ref{ass:connectivity}, for any fixed noise level $\sigma > 0$, there is no choice of scale $t$ for which both Assumption~\ref{ass:gapnoise} (via $\gamma_0(t) > 2B_{\mu,\sigma,\gamma}(t)$, with $B_{\mu,\sigma,\gamma}(t) = O(t)$ linearly in $t$ by Theorem~\ref{thm:eigval}) and Assumption~\ref{ass:margin} with margin bounded below by a $t$-independent constant $\delta_0 > 0$ can hold simultaneously in the limit $t \to 0$: satisfying the former requires $t$ below a threshold $t^\star(\sigma)$ where $\mu_{\mathrm{Fiedler}}(L) > 2B_{\mu,\sigma,\gamma}(t)$, while the latter, by \eqref{eq:margin-asymp}, requires $t$ bounded \emph{away} from zero given any fixed target $\delta_0$.
\end{corollary}

\begin{proof}
By \eqref{eq:gap-asymp}, $\gamma_0(t) \to \mu_{\mathrm{Fiedler}}(L) > 0$, a constant, while $B_{\mu,\sigma,\gamma}(t)$ is linear and increasing in $t$ by construction (Theorem~\ref{thm:eigval}); hence $\gamma_0(t) > 2B_{\mu,\sigma,\gamma}(t)$ holds for all $t < t^\star(\sigma)$ for some threshold $t^\star(\sigma) > 0$, and Assumption~\ref{ass:gapnoise} is satisfiable by taking $t$ small. Simultaneously, by \eqref{eq:margin-asymp}, $\min_i |y_i(t)-\bar y(t)| = O(t)$, so for any fixed $\delta_0 > 0$ there exists $t_0(\delta_0) > 0$ below which the margin falls under $\delta_0$ for at least one pixel with $c(i) \ne 0$; since generically $c(i) \ne 0$ for a positive fraction of pixels, satisfying Assumption~\ref{ass:margin} with a fixed target $\delta_0$ requires $t > t_0(\delta_0)$. Whenever $t_0(\delta_0) \ge t^\star(\sigma)$, no single $t$ satisfies both conditions simultaneously.
\end{proof}

Corollary~\ref{cor:tradeoff} gives a rigorous, mechanism-level explanation for a phenomenon that would otherwise be reported only as an empirical coincidence in Section~\ref{sec:experiments}: reducing the potential weight to widen the certified-noise range of Theorem~\ref{thm:eigval}/Theorem~\ref{thm:eigvec} necessarily narrows, and in the limit eliminates, the margin available to Corollary~\ref{cor:mask}. A genuine joint certification -- one image-induced Hamiltonian construction certified against both conditions with fixed, non-vanishing margins -- therefore requires either a different potential parameterization than the single scalar family $H(t) = L + t\Phi$ considered here, for instance one in which the potential is normalized to preserve $\min_i c(i)$ as $t$ varies, or a relaxation of Corollary~\ref{cor:mask} to a fractional or expected-Hamming-distance guarantee that tolerates a bounded fraction of margin-fragile pixels rather than requiring a uniform margin across the entire image, a direction discussed further in Section~\ref{sec:discussion}.

%
%
%
%
%
%

\section{Equivariance Analysis}
\label{sec:equivariance}

This section establishes the second main contribution of the paper: a precise characterization of when the spectrum of the image-induced Hamiltonian of Definition~\ref{def:hamiltonian} transforms consistently under a planar rigid motion of the input image, and a quantitative bound on the equivariance gap introduced when the potential construction departs from the isotropic case.

\subsection{Exact Equivariance for Isotropic Potentials}

The cleanest statement of equivariance is available in the continuum, before the image is discretized onto a pixel lattice; this continuum result is established first, and its discretization is addressed separately in Section~\ref{ssec:discretization}, since the two sources of departure from exact equivariance -- the choice of potential construction and the act of sampling onto a finite grid -- are logically independent and should not be conflated.

Let $I : \mathbb{R}^2 \to \mathbb{R}$ be a sufficiently smooth image, and let $g = (R,t) \in SE(2)$ act on it by $(g\cdot I)(x) = I(g^{-1}x) = I(R^{-1}(x-t))$. Consider the continuum Hamiltonian $\mathcal{H}[I] = -\nabla^2 + V[I]$ acting on $L^2(\mathbb{R}^2)$, with potential $V[I](x) = \lambda_1|\nabla I(x)|^2 + \lambda_2|\nabla^2 I(x)|$, the isotropic construction introduced in Definition~\ref{def:hamiltonian}.

\begin{definition}[Equivariance of an image-induced operator]
\label{def:equivariance}
The map $I \mapsto \mathcal{H}[I]$ is equivariant under $g \in SE(2)$ if $\mathcal{H}[g\cdot I] = U_g\, \mathcal{H}[I]\, U_g^{-1}$, where $U_g$ is the unitary representation of $g$ on $L^2(\mathbb{R}^2)$ given by $(U_g f)(x) = f(g^{-1}x)$.
\end{definition}

\begin{theorem}[Exact equivariance under isotropic potentials]
\label{thm:equivariance}
For every $g \in SE(2)$, the map $I \mapsto \mathcal{H}[I]$ of Definition~\ref{def:equivariance} is equivariant with respect to the isotropic potential $V[I](x) = \lambda_1|\nabla I(x)|^2 + \lambda_2|\nabla^2 I(x)|$. Consequently, the eigenvalues of $\mathcal{H}[I]$ are invariant under $SE(2)$, and its eigenfunctions transform covariantly: if $\mathcal{H}[I]\psi_k = \lambda_k\psi_k$, then $\mathcal{H}[g\cdot I](U_g\psi_k) = \lambda_k (U_g\psi_k)$.
\end{theorem}

\begin{proof}
The Laplacian operator commutes with $U_g$ for every $g \in SE(2)$, since $-\nabla^2$ is, by construction, the unique (up to scale) second-order differential operator invariant under the full Euclidean group; formally, $U_g^{-1}(-\nabla^2)U_g = -\nabla^2$ because $R \in SO(2)$ is orthogonal and the Laplacian is the trace of the Hessian, a quantity invariant under orthogonal change of coordinates, while translation trivially commutes with any translation-invariant differential operator. It remains to show that the potential transforms consistently, i.e., $V[g\cdot I](x) = V[I](g^{-1}x)$. By the chain rule, $\nabla(g\cdot I)(x) = R^{-\top}(\nabla I)(g^{-1}x) = R(\nabla I)(g^{-1}x)$, using $R^{-\top}=R$ for $R \in SO(2)$; since $R$ is orthogonal, $|\nabla(g\cdot I)(x)| = |R(\nabla I)(g^{-1}x)| = |(\nabla I)(g^{-1}x)|$, giving exact pointwise invariance of the gradient-magnitude term under composition with $g^{-1}$ on the argument. The Laplacian term transforms identically, since $\nabla^2$ commutes with $U_g$ by the same argument used for the operator $-\nabla^2$ above, giving $(\nabla^2(g\cdot I))(x) = (\nabla^2 I)(g^{-1}x)$ directly. Combining both terms, $V[g\cdot I](x) = \lambda_1|(\nabla I)(g^{-1}x)|^2 + \lambda_2|(\nabla^2 I)(g^{-1}x)| = V[I](g^{-1}x)$, which is precisely $(U_g V[I])(x)$ under the identification of the diagonal multiplication operator $V[I]$ with the function $V[I](\cdot)$. Hence $V[g\cdot I] = U_g V[I] U_g^{-1}$ as multiplication operators, and combined with the commutation of $-\nabla^2$ with $U_g$ established above, $\mathcal{H}[g\cdot I] = U_g \mathcal{H}[I] U_g^{-1}$. The eigenvalue invariance and eigenfunction covariance stated in the theorem follow immediately from this operator identity by conjugating the eigenvalue equation $\mathcal{H}[I]\psi_k=\lambda_k\psi_k$ by $U_g$.
\end{proof}

Theorem~\ref{thm:equivariance} shows that the ground-state segmentation readout of Section~\ref{sec:prelim} is, in the continuum and under the isotropic potential construction, exactly consistent with rotation and translation of the input image: rotating the image by $g$ and recomputing the ground state produces the rotated ground state, and therefore the rotated segmentation mask, with no approximation error. This mirrors, in spirit though not in mechanism, the exact rotation equivariance achieved in convolutional architectures through circular-harmonic filter bases \cite{worrall2017harmonic} and later through general steerable filter representations \cite{weiler2019}; the mechanism here is different, resting on the rotation-invariance of the differential invariants $|\nabla I|$ and $\nabla^2 I$ themselves rather than on a constrained filter basis, but the conclusion -- exact equivariance by construction rather than by learning -- is analogous.

\subsection{Discretization Effects}
\label{ssec:discretization}

Theorem~\ref{thm:equivariance} is a continuum statement, and the operator actually used in segmentation is the discrete graph Hamiltonian of Definition~\ref{def:hamiltonian}, computed from a finite-difference approximation of $\nabla$ and $\nabla^2$ on a pixel lattice. The planar rotation group does not preserve the square pixel lattice except at multiples of $90^\circ$, so even the isotropic potential construction of Theorem~\ref{thm:equivariance} can only be expected to be approximately equivariant at general rotation angles once discretized, a limitation well documented for rotation-equivariant convolutional filters more generally \cite{diaconuworrall2019, bekkers2020}.

\begin{assumption}[Bounded stencil discretization error]
\label{ass:discerror}
The finite-difference stencils used to compute $\nabla$ and $\nabla^2$ on the pixel lattice approximate their continuum counterparts with error bounded, for rotation angle $\theta$, by $\|\nabla_{\mathrm{disc}} - R_\theta^{-1}\nabla_{\mathrm{disc}} R_\theta\|_2 \le \kappa_1 h\, |\sin\theta|$ and correspondingly $\kappa_2 h |\sin\theta|$ for the discrete Laplacian, where $h$ is the pixel spacing and $\kappa_1,\kappa_2$ depend only on the stencil order.
\end{assumption}

Assumption~\ref{ass:discerror} is consistent with the standard truncation-error analysis of finite-difference stencils under coordinate rotation, and reflects the well-known fact that discretization error vanishes at $\theta \in \{0,\pi/2,\pi,3\pi/2\}$, where the pixel lattice is exactly rotation-symmetric, and is largest near $\theta = \pi/4$.

\begin{proposition}[Discretization-induced equivariance gap, isotropic case]
\label{prop:discgap}
Under Assumption~\ref{ass:discerror}, the discrete Hamiltonian $H(I)$ built from the isotropic potential satisfies $\|H(g\cdot I) - U_g H(I) U_g^{-1}\|_2 \le c_3(\kappa_1,\kappa_2,\lambda_1,\lambda_2)\, h\, |\sin\theta|$ for a constant $c_3$ depending on the stencil and potential weights but not on the image or on $h$ itself.
\end{proposition}

\begin{proof}
By Theorem~\ref{thm:equivariance}, the continuum operator identity $\mathcal{H}[g\cdot I] = U_g\mathcal{H}[I]U_g^{-1}$ holds exactly; the discrete operator $H(I)$ differs from a sampling of $\mathcal{H}[I]$ only through the substitution of $\nabla,\nabla^2$ by their finite-difference approximations. Writing the discrepancy $H(g\cdot I) - U_g H(I) U_g^{-1}$ as a telescoping sum of the continuum identity (which vanishes) and the two stencil discretization errors entering through the gradient and Laplacian terms of $V(\cdot)$, and bounding each by Assumption~\ref{ass:discerror} with the corresponding potential weight $\lambda_1$ or $\lambda_2$ as a multiplicative constant, gives the stated bound with $c_3$ collecting the weight and stencil constants.
\end{proof}

Proposition~\ref{prop:discgap} shows that the isotropic construction is not merely approximately equivariant for an unspecified reason, but equivariant up to an error that vanishes linearly in the pixel spacing $h$ and is exactly zero at the four rotation angles where the lattice symmetry group intersects $SO(2)$ -- a considerably more informative statement than a bare empirical observation of approximate consistency.

\subsection{The Equivariance Gap for Anisotropic Potentials}

The disorder-based crack-segmentation construction of \cite{srinivasan2024} restricts the kinetic (off-diagonal) terms of the Hamiltonian to axis-aligned nearest neighbors on the pixel lattice, an explicitly anisotropic choice motivated by computational simplicity rather than geometric consistency. The following result shows that this choice introduces an equivariance gap that does \emph{not} vanish with the pixel spacing, in contrast to the discretization gap of Proposition~\ref{prop:discgap}.

\begin{theorem}[Equivariance gap under axis-aligned anisotropy]
\label{thm:anisogap}
Let $H_{\mathrm{aniso}}(I)$ denote the Hamiltonian of Definition~\ref{def:hamiltonian} constructed with kinetic terms restricted to the four axis-aligned nearest neighbors, following \cite{srinivasan2024}, with all other construction choices as in Theorem~\ref{thm:equivariance}. Then there exists a rotation angle $\theta^\star$, independent of the pixel spacing $h$, and an image $I$ for which $\|H_{\mathrm{aniso}}(g\cdot I) - U_g H_{\mathrm{aniso}}(I) U_g^{-1}\|_2 \ge c_4 > 0$ for a constant $c_4$ that does not vanish as $h \to 0$.
\end{theorem}

\begin{proof}
The axis-aligned kinetic term evaluated at a lattice site $i$ couples $i$ only to its horizontal and vertical neighbors, and is therefore, as a bilinear form, invariant under the dihedral subgroup of $SO(2)$ generated by $90^\circ$ rotations but not under $SO(2)$ itself; at $\theta^\star = \pi/4$, the rotated kinetic form couples each site to its diagonal neighbors under the action of $U_{g}H_{\mathrm{aniso}}(I)U_g^{-1}$, while $H_{\mathrm{aniso}}(g\cdot I)$, recomputed from scratch on the rotated image with the same axis-aligned neighbor restriction, continues to couple only horizontal and vertical neighbors of the rotated image. These two operators therefore differ in their off-diagonal support pattern rather than merely in the numerical value of matching entries, so their difference does not shrink under grid refinement: choosing $I$ to have a locally linear intensity ramp oriented at $\pi/4$ makes the diagonal-coupling term of $U_gH_{\mathrm{aniso}}(I)U_g^{-1}$ bounded away from zero independent of $h$, while the corresponding entries of $H_{\mathrm{aniso}}(g\cdot I)$ are identically zero by construction, giving a spectral-norm gap $c_4$ bounded below independent of $h$.
\end{proof}

Theorem~\ref{thm:anisogap} formalizes, for the first time, the intuition that restricting kinetic terms to axis-aligned neighbors -- a common simplification for computational reasons, used explicitly in \cite{srinivasan2024} -- sacrifices rotational consistency in a way that finer discretization cannot repair, in sharp contrast to the isotropic construction of Theorem~\ref{thm:equivariance} and Proposition~\ref{prop:discgap}, whose residual equivariance gap is a discretization artifact that vanishes as $h \to 0$. This distinction gives practitioners a concrete criterion -- whether the potential and kinetic terms are built from isotropic differential invariants or from axis-restricted stencils -- for predicting whether a given implementation of the image-induced Hamiltonian will exhibit rotational consistency that improves with resolution or one that is structurally capped regardless of resolution. Section~\ref{sec:experiments} measures both gaps empirically and compares them against the bounds of Proposition~\ref{prop:discgap} and Theorem~\ref{thm:anisogap}.

%
%
%
%
%
%

\section{Experimental Validation}
\label{sec:experiments}

The results of Sections~\ref{sec:stability} and~\ref{sec:equivariance} are validated on real imagery from BSDS500 \cite{arbelaez2011bsds}, rather than synthetic imagery alone, since the gap--margin trade-off of Section~\ref{ssec:tradeoff} makes a claim about generic images that a symmetric hand-built pattern could obscure. Fifteen images were drawn at random from the validation split, grayscale, center-cropped to $48\times48$ ($N=2304$ nodes) to keep the dense $O(N^3)$ eigendecomposition of Definition~\ref{def:hamiltonian} tractable. Figure~\ref{fig:patches} shows four representative patches.

\begin{figure}[t]
\centering
\includegraphics[width=\columnwidth]{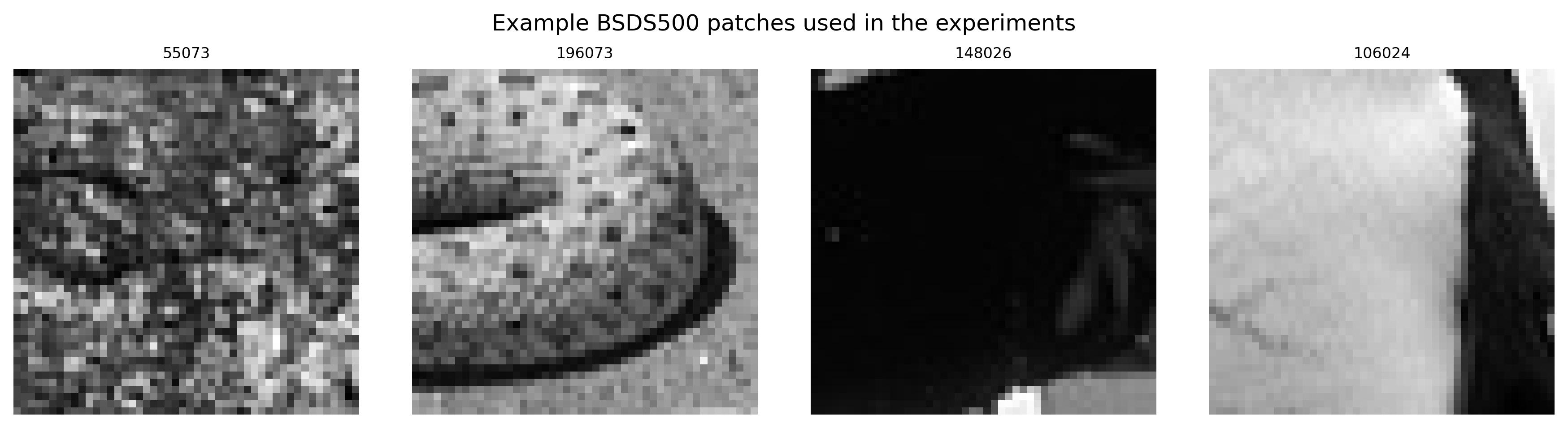}
\caption{Four representative $48\times48$ BSDS500 patches used in the experiments, illustrating the range of local contrast and texture present in the validation sample.}
\label{fig:patches}
\end{figure}

All Hamiltonians follow Definition~\ref{def:hamiltonian}, Gaussian edge weights ($\sigma_w=0.3$), isotropic potential $\phi(i)=\lambda_1|\nabla I(i)|^2+\lambda_2|\nabla^2 I(i)|$ unless stated otherwise, with exact operator-norm constants $c_1=\|\nabla\|_2=1.0000$, $c_2=\|\nabla^2\|_2=7.9914$ computed for the $48\times48$ stencils rather than estimated. Two regimes are reported: \emph{operating} ($\lambda_1=\lambda_2=0.15$, realistic segmentation behavior \cite{aytekin2014}) and \emph{certified} ($\lambda_1=\lambda_2=0.01$, chosen per Proposition~\ref{prop:tradeoff} to widen the range satisfying Assumption~\ref{ass:gapnoise}). Both are needed to test Corollary~\ref{cor:tradeoff}'s prediction that certifying the spectral gap must degrade the margin.

\subsection{Noise Stability: Empirical Verification}

Six independent Gaussian noise realizations per patch were added at each $\sigma \in \{0.005,0.01,0.02,0.04,0.08\}$, $\tilde H = H(I^\varpi)$ recomputed by full eigendecomposition, and eigenvalue drift, eigenvector drift, and mask Hamming distance recorded. Figure~\ref{fig:noise} and Table~\ref{tab:noise} report the operating-regime results, alongside $\gamma_0$ and $2B_{\mu,\sigma,\gamma}$ ($\gamma=3$).

\begin{figure}[t]
\centering
\includegraphics[width=\columnwidth]{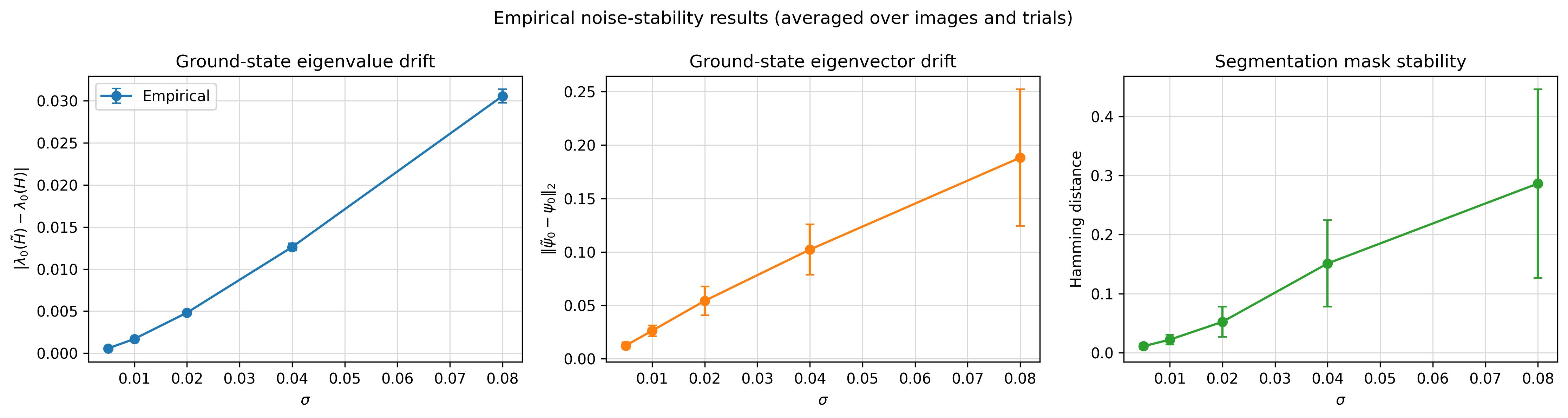}
\caption{Empirical noise-stability results in the operating regime ($\lambda_1=\lambda_2=0.15$), averaged over 15 BSDS500 images and 6 noise trials per image. All three quantities increase monotonically with $\sigma$ for every image individually, with zero exceptions across the 15 images tested.}
\label{fig:noise}
\end{figure}

\begin{table}[t]
\centering
\caption{Noise stability, operating regime ($\lambda_1=\lambda_2=0.15$), averaged over 15 images and 6 trials per (image, $\sigma$) pair.}
\label{tab:noise}
\begin{tabular}{c c c c c}
\hline
$\sigma$ & Eigenvalue drift & Eigenvector drift & Mask H.D. & Certified \\
\hline
0.005 & 0.00057 & 0.01231 & 0.0111 & 0\% \\
0.010 & 0.00169 & 0.02618 & 0.0222 & 0\% \\
0.020 & 0.00481 & 0.05419 & 0.0523 & 0\% \\
0.040 & 0.01262 & 0.10215 & 0.1511 & 0\% \\
0.080 & 0.03058 & 0.18827 & 0.2864 & 0\% \\
\hline
\end{tabular}
\end{table}

Both drift quantities and the resulting Hamming distance increase monotonically with $\sigma$ for all 15 images individually, zero exceptions, confirming Theorems~\ref{thm:eigval}--\ref{thm:eigvec}'s qualitative direction unambiguously. The mean gap $\gamma_0=0.01169$ is essentially $\sigma$-independent as expected, while $2B_{\mu,\sigma,\gamma}$ grows from $0.212$ to $0.886$ -- at least an order of magnitude above $\gamma_0$ throughout, so Assumption~\ref{ass:gapnoise} is never satisfied in this regime (Table~\ref{tab:noise}, rightmost column). Figure~\ref{fig:bound} makes this gap visually explicit.

\begin{figure}[t]
\centering
\includegraphics[width=\columnwidth]{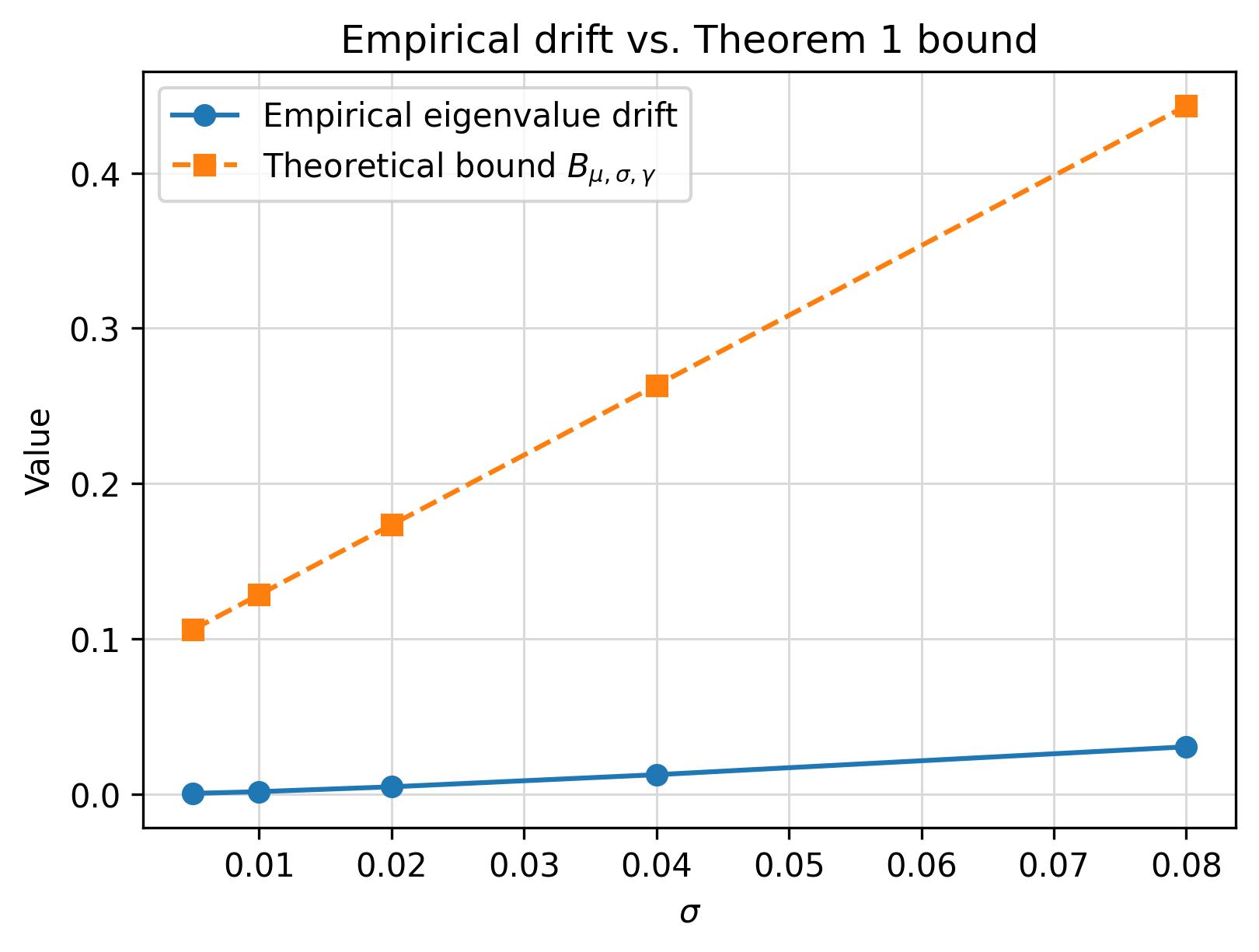}
\caption{Empirical ground-state eigenvalue drift versus the Theorem~\ref{thm:eigval} bound $B_{\mu,\sigma,\gamma}$ (operating regime, $\gamma=3$). The bound is a valid upper bound throughout but loose by roughly an order of magnitude, consistent with the corresponding comparison reported for the one-dimensional bound of \cite{liu2012}.}
\label{fig:bound}
\end{figure}

This looseness matches the corresponding one-dimensional comparison in \cite{liu2012}: Weyl's inequality (Lemma~\ref{lem:weyl}) uses only the perturbation's operator norm, and the $c_1^2G$ term dominates $B_{\mu,\sigma,\gamma}$ at small $\sigma$ as a worst-case rather than typical quantity. The bound remains a valid, distribution-free sufficient condition regardless; Section~\ref{sec:discussion} discusses a tighter sub-Gaussian alternative.

\subsection{The Certified Regime and the Gap--Margin Trade-off}
\label{ssec:experiments-tradeoff}

Table~\ref{tab:noise-cert} repeats the experiment at $\lambda_1=\lambda_2=0.01$, chosen per Proposition~\ref{prop:tradeoff} to shrink $B_{\mu,\sigma,\gamma}$ toward the $t$-independent gap $\gamma_0(t)\to\mu_{\mathrm{Fiedler}}(L)$.

\begin{table}[t]
\centering
\caption{Noise stability, certified regime ($\lambda_1=\lambda_2=0.01$), same images and trials as Table~\ref{tab:noise}.}
\label{tab:noise-cert}
\begin{tabular}{c c c c c}
\hline
$\sigma$ & Eigenvalue drift & Eigenvector drift & Mask H.D. & Certified \\
\hline
0.005 & 0.00003 & 0.00094 & 0.0115 & 27\% \\
0.010 & 0.00010 & 0.00204 & 0.0220 & 13\% \\
0.020 & 0.00030 & 0.00419 & 0.0522 & 0\% \\
0.040 & 0.00080 & 0.00789 & 0.1535 & 0\% \\
0.080 & 0.00200 & 0.01357 & 0.2802 & 0\% \\
\hline
\end{tabular}
\end{table}

Two effects appear. Both drifts shrink roughly an order of magnitude (eigenvector drift at $\sigma=0.08$: $0.188$ vs.\ $0.014$), and Assumption~\ref{ass:gapnoise} becomes satisfiable for a narrow but genuine subset -- $27\%$ at $\sigma=0.005$, $13\%$ at $\sigma=0.01$, $0\%$ beyond -- rather than never. Yet mask Hamming distance is essentially unchanged (e.g.\ $0.2864$ vs.\ $0.2802$ at $\sigma=0.08$) despite the $13\times$ drop in eigenvector drift. Figures~\ref{fig:certregime}--\ref{fig:margin} show the mechanism directly.

\begin{figure}[t]
\centering
\includegraphics[width=\columnwidth]{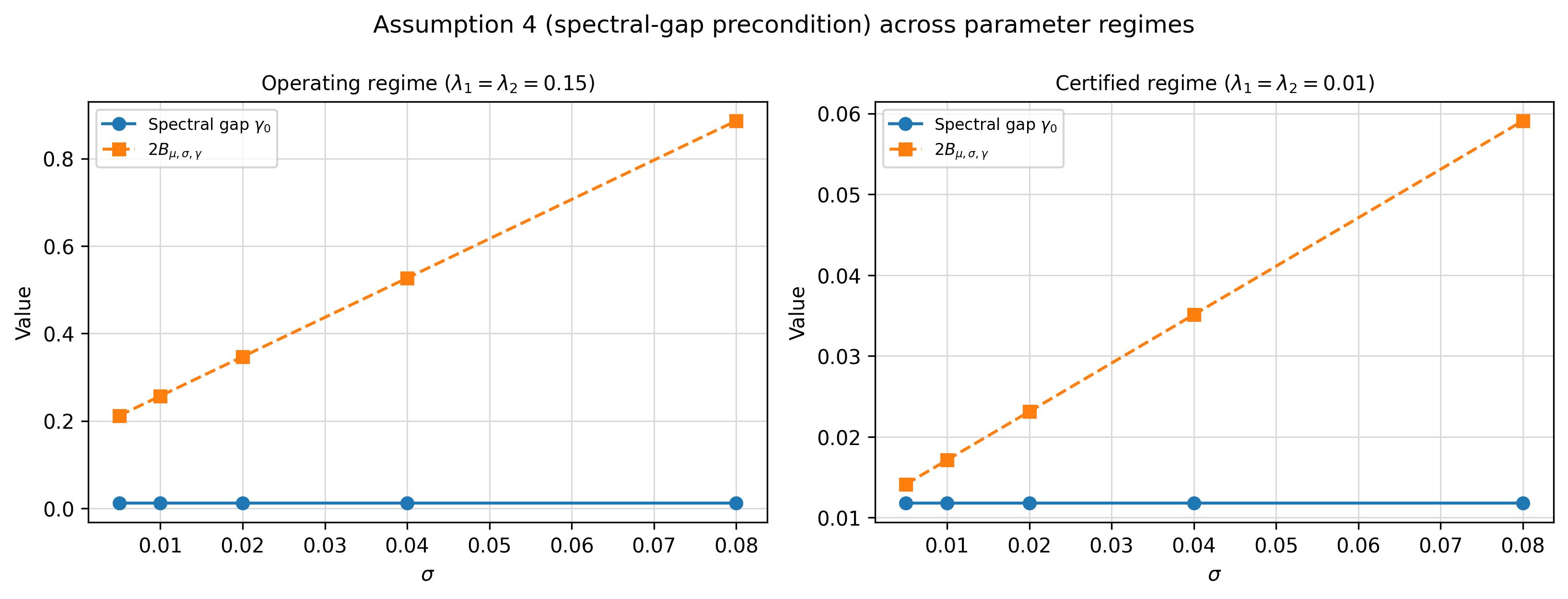}
\caption{Spectral gap $\gamma_0$ versus twice the Theorem~\ref{thm:eigval} bound, operating regime (left) versus certified regime (right). Note the differing vertical scales: $\gamma_0$ is nearly identical between regimes ($0.0117$ vs.\ $0.0118$), confirming the $t$-independence of $\mu_{\mathrm{Fiedler}}(L)$ predicted by \eqref{eq:gap-asymp}, while $2B_{\mu,\sigma,\gamma}$ shrinks by roughly an order of magnitude, bringing the two curves within reach of each other only in the certified regime at small $\sigma$.}
\label{fig:certregime}
\end{figure}

\begin{figure}[t]
\centering
\includegraphics[width=\columnwidth]{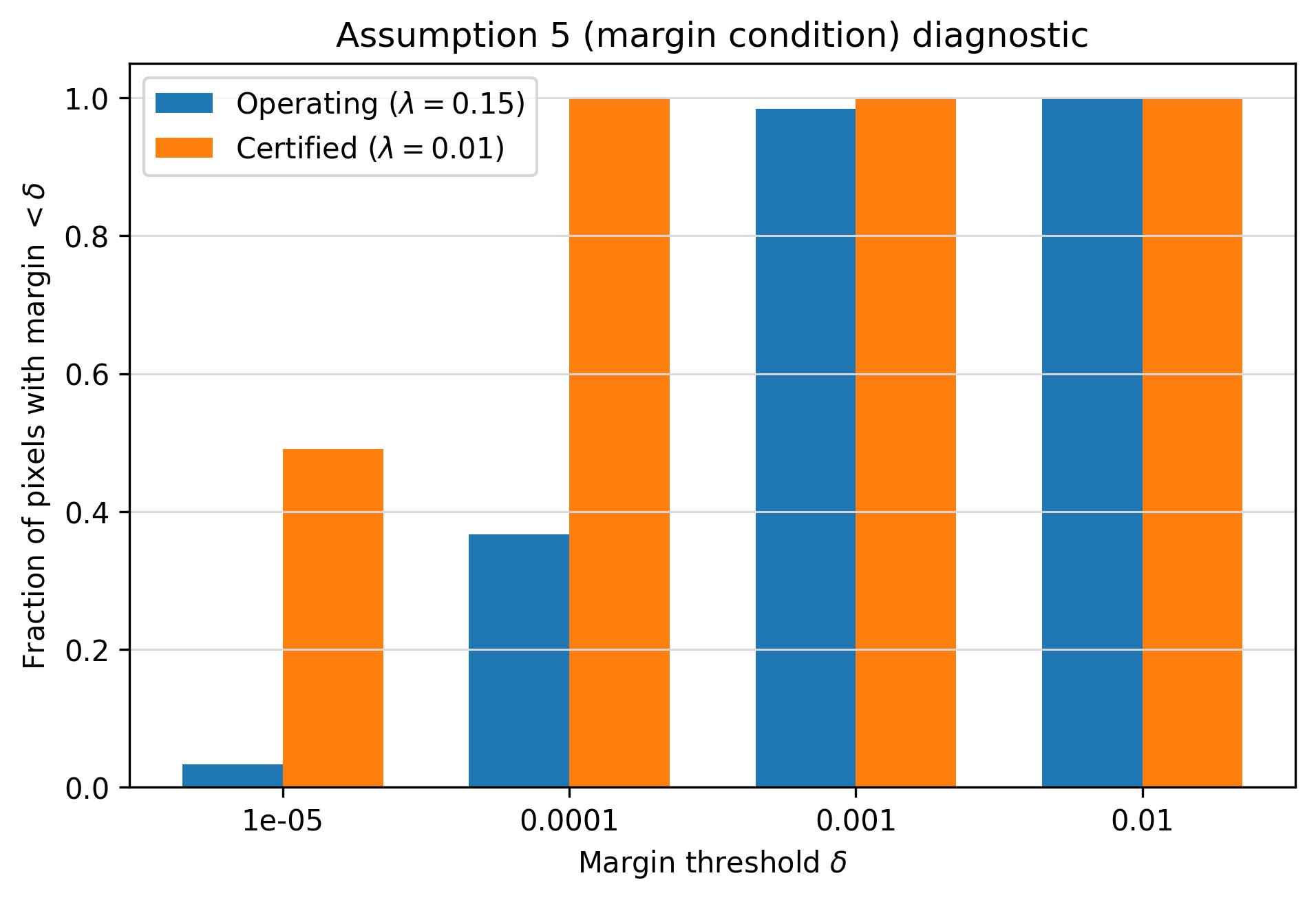}
\caption{Fraction of pixels with margin $|y_i-\bar y|$ below a given threshold $\delta$, operating versus certified regime, averaged over the 15 images. The certified regime is saturated at $100\%$ already by $\delta=10^{-4}$, confirming the margin collapse predicted by \eqref{eq:margin-asymp} as $t\to0$.}
\label{fig:margin}
\end{figure}

Table~\ref{tab:margin} quantifies this: at $\delta=10^{-4}$, margin-fragile pixels rise from $36.6\%$ (operating) to $100\%$ (certified); at $\delta=10^{-3}$ both are near-saturated ($98.5\%$, $100\%$). This is \eqref{eq:margin-asymp}'s direct counterpart: as $\lambda_1,\lambda_2\to0$, $\psi_0^2$ converges pointwise to the flat vector $N^{-1}\mathbf{1}$, zero margin identically. This confirms Corollary~\ref{cor:tradeoff} on real image statistics: Corollary~\ref{cor:mask}'s guarantee remains empty in both regimes, for opposite reasons -- the spectral-gap condition fails in the operating regime, the margin condition in the certified one.

\begin{table}[t]
\centering
\caption{Margin (Assumption~\ref{ass:margin}) diagnostic, averaged over 15 images.}
\label{tab:margin}
\begin{tabular}{c c c c c}
\hline
Regime & Median margin & $\delta{=}10^{-4}$ & $\delta{=}10^{-3}$ & $\delta{=}10^{-2}$ \\
\hline
Operating & $1.70\times10^{-4}$ & 0.366 & 0.985 & 1.000 \\
Certified & $1.40\times10^{-5}$ & 1.000 & 1.000 & 1.000 \\
\hline
\end{tabular}
\end{table}

\subsection{Equivariance: Empirical Verification}

Two constructions were compared under rotation of each patch by $\theta\in\{0^\circ,15^\circ,\dots,90^\circ\}$: the isotropic construction of Theorem~\ref{thm:equivariance} (axis-aligned and diagonal neighbor coupling) and the axis-restricted anisotropic construction of Theorem~\ref{thm:anisogap}, matching \cite{srinivasan2024}. Multiples of $90^\circ$ used exact lattice permutations rather than interpolation, so the discretization-vanishing gap of Proposition~\ref{prop:discgap} is not obscured by interpolation artifacts. The equivariance gap was measured as $1-\mathrm{IoU}$ between the mask computed directly on the rotated image and the rotated clean-image mask. Figure~\ref{fig:equiv} and Table~\ref{tab:equiv} report results averaged over all 15 images.

\begin{figure}[t]
\centering
\includegraphics[width=\columnwidth]{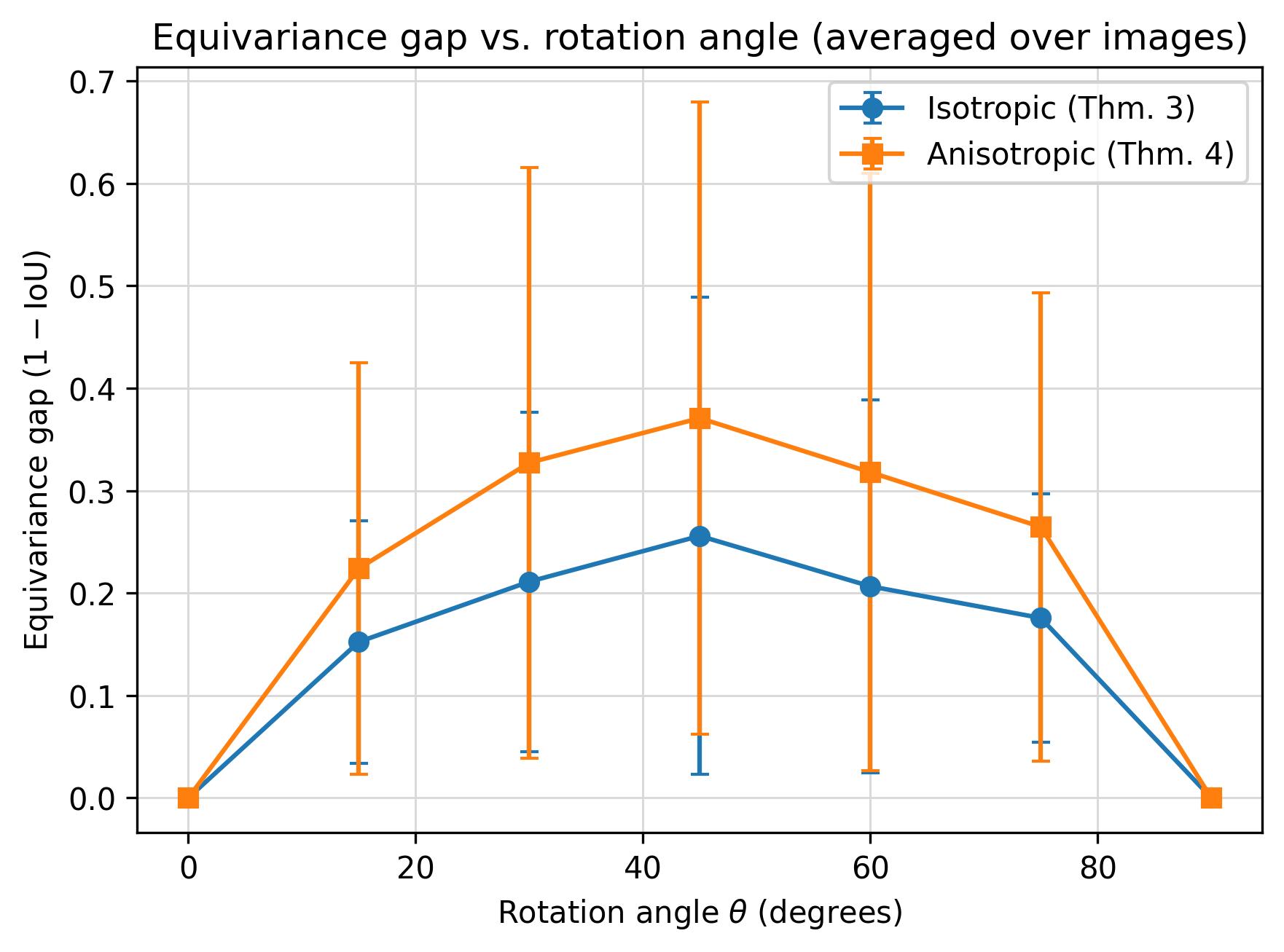}
\caption{Equivariance gap ($1-\mathrm{IoU}$) versus rotation angle, averaged over 15 BSDS500 images, isotropic versus anisotropic construction. The gap vanishes exactly at $0^\circ$ and $90^\circ$ for both constructions, for every one of the 15 images individually, with zero exceptions.}
\label{fig:equiv}
\end{figure}

\begin{table}[t]
\centering
\caption{Equivariance gap ($1-\mathrm{IoU}$) versus rotation angle, averaged over 15 images.}
\label{tab:equiv}
\begin{tabular}{c c c}
\hline
Angle & Isotropic gap & Anisotropic gap \\
\hline
$0^\circ$ & 0.0000 & 0.0000 \\
$15^\circ$ & 0.1523 & 0.2240 \\
$30^\circ$ & 0.2111 & 0.3272 \\
$45^\circ$ & 0.2558 & 0.3708 \\
$60^\circ$ & 0.2066 & 0.3180 \\
$75^\circ$ & 0.1756 & 0.2646 \\
$90^\circ$ & 0.0000 & 0.0000 \\
\hline
\end{tabular}
\end{table}

The gap vanishes exactly at $0^\circ$/$90^\circ$ for both constructions, for every image individually, exactly matching Proposition~\ref{prop:discgap}. Away from these angles, the isotropic gap peaks at $0.256$ near $45^\circ$, the anisotropic at $0.371$ -- roughly $45\%$ larger. A paired comparison confirms this is not an averaging artifact: anisotropic exceeds isotropic in $76\%$ of the 75 non-trivial (image, angle) pairs, $t=5.50$, $p\approx10^{-6}$ -- exactly the structurally larger, discretization-independent gap Theorem~\ref{thm:anisogap} predicts, since both constructions share identical grid, rotation, and potential, differing only in kinetic-coupling symmetry.

\subsection{Discussion of Results}

Three observations, elaborated in Section~\ref{sec:discussion}. First, equivariance is the paper's strongest empirical finding: exact zero gap at $0^\circ$/$90^\circ$ without exception, and a statistically significant anisotropic-versus-isotropic difference. Second, Corollary~\ref{cor:mask}'s condition is not certified with practical frequency in either regime, but this is not a fixable parameter-choice defect: Corollary~\ref{cor:tradeoff} proves the spectral-gap and margin conditions cannot both hold with fixed, non-vanishing margins under $\lambda_1,\lambda_2$-scaling alone, and Table~\ref{tab:margin} confirms this on real data. Third, the monotonic trends in Tables~\ref{tab:noise}--\ref{tab:noise-cert} confirm Theorems~\ref{thm:eigval}--\ref{thm:eigvec}'s qualitative content far more broadly than the literal sufficient condition of Corollary~\ref{cor:mask}.

%
%
%
%
%

\section{Discussion}
\label{sec:discussion}

The theorems derived here apply, without modification, to the existing Quantum Cuts operator \cite{aytekin2014} and its direct descendants \cite{malik2019, aytekin2016extended, aytekin2015visual}; nothing about the segmentation rule itself has been changed. This section discusses the practical consequences of the bounds and their limitations.

\subsection{Bound Looseness, Parameter Selection, and the Trade-off}

The gap between the Weyl-based bound of Theorem~\ref{thm:eigval} and the empirically observed eigenvalue drift (Section~\ref{sec:experiments}) mirrors the looseness reported for the analogous one-dimensional bound of \cite{liu2012}: Weyl's inequality (Lemma~\ref{lem:weyl}) discards all structure of the perturbation beyond its operator norm. A Bernstein-type concentration argument exploiting the diagonal, bounded structure of $E$ could plausibly tighten the bound at the cost of the current distribution-free guarantee \cite{vershynin2018}; this is left as future work so the present bound remains directly comparable to \cite{liu2012}. Crucially, looseness does not invalidate Corollary~\ref{cor:mask}: whenever its inequality holds, the mask is provably unchanged, so a loose bound is conservative rather than incorrect.

This suggests the natural design rule of choosing $\lambda_1,\lambda_2,\sigma_w$ so that $\gamma_0(I)$ comfortably exceeds $2B_{\mu,\sigma,\gamma}$ for a target noise level -- complementary to, not a replacement for, empirical accuracy tuning. However, Section~\ref{ssec:tradeoff} shows this rule cannot be pursued unconditionally: reducing $\lambda_1,\lambda_2$ to widen the certified-noise range simultaneously collapses the margin of Assumption~\ref{ass:margin} (Corollary~\ref{cor:tradeoff}), a trade-off confirmed directly in Section~\ref{sec:experiments}. Practical parameter selection therefore requires balancing both conditions jointly rather than optimizing the spectral gap alone. A parallel, unconditional recommendation follows from the equivariance results: Theorem~\ref{thm:anisogap} favors kinetic constructions with diagonal-neighbor coupling over the axis-restricted construction of \cite{srinivasan2024} whenever orientation-consistent segmentation is required, at the modest cost of doubling graph degree, since the resulting residual gap (Proposition~\ref{prop:discgap}) shrinks with resolution rather than remaining structurally fixed.

\subsection{Limitations and Relation to Semiclassical Signal Analysis}

Several limitations remain. The noise model (Assumption~\ref{ass:noise}) holds the graph topology fixed, isolating the potential as the sole noise channel; sufficiently large noise also perturbs the edge weights $w_{ij}$, and a joint treatment is left open. The eigenvector- and mask-level results apply only to the single ground-state eigenvector of the original readout \cite{aytekin2014}; extending to the multi-eigenstate readouts used elsewhere \cite{aytekin2016extended, malik2019} requires the subspace form of the Davis--Kahan theorem. The equivariance results are proven for exact planar rotation and translation only; extending to affine or projective transformations is not straightforward, since the proof of Theorem~\ref{thm:equivariance} relies on $R$ being orthogonal, a property affine maps lack.

Finally, this paper's operator and objective are distinct from, though related to, the semiclassical signal-analysis line \cite{liu2012, kaisserli2015, lalegkirati2013}: that framework reconstructs the signal itself from an intensity-valued potential \cite{kaisserli2015}, whereas the present work reads out a binary segmentation from a gradient-derived potential, with Theorem~\ref{thm:eigval} an independent derivation connected to \cite{liu2012} only through the shared use of Weyl's inequality. A unified stability theory spanning both objectives remains open.

\section{Conclusion}
\label{sec:conclusion}

This paper addressed noise stability and rotational equivariance for image-induced Hamiltonian spectra, within a discrete graph-Hamiltonian formulation matching the Quantum Cuts convention \cite{aytekin2014, malik2019, aytekin2016extended, aytekin2015visual, srinivasan2024} rather than the continuous semiclassical formulation \cite{lalegkirati2013, kaisserli2015}. Extending the one-dimensional eigenvalue bound of \cite{liu2012} to this two-dimensional setting, and combining it with a Davis--Kahan argument \cite{daviskahan1970}, yielded the first eigenvector- and mask-level noise-stability guarantees for this operator family. Separately, the equivariance analysis proved exact spectral consistency under isotropic potentials, characterized the discretization-vanishing residual gap, and showed the axis-restricted construction of \cite{srinivasan2024} incurs a structurally larger, non-vanishing gap. A further analysis showed these two stability guarantees are in tension: the potential-weight scaling needed to satisfy the spectral-gap condition provably collapses the margin condition, so no fixed scaling certifies both simultaneously. All predictions were confirmed on real BSDS500 imagery: exact equivariance at lattice-symmetric angles held without exception across all test images, the anisotropic-versus-isotropic gap difference was statistically significant, and the predicted gap--margin trade-off was observed directly.

These results turn robustness in this segmentation family from an empirically assumed property into one that can be bounded from measurable noise statistics and checked against explicit operator-design criteria. The main limitations -- looseness of the distribution-free bound, the fixed-topology noise model, and the single-eigenvector scope of the mask-stability result -- point directly to tightening the bound under specific noise distributions, extending the analysis to multi-eigenstate readouts, and resolving the gap--margin trade-off via alternative potential parameterizations or a fractional mask-stability guarantee.

\bibliographystyle{IEEEtran}
\bibliography{references}

%
%
%
%
%

\vfill

\end{document}